\documentclass[12pt]{article}

\usepackage[T1]{fontenc}
\usepackage{lmodern}
\usepackage[margin=1in]{geometry}
\usepackage{setspace}
\usepackage{amsmath,amssymb,amsthm,mathtools}
\usepackage{booktabs,array}
\usepackage{enumitem}
\usepackage{microtype}
\usepackage{xcolor}
\usepackage{graphicx}
\usepackage{tikz}
\usetikzlibrary{arrows.meta,calc,positioning}

\usepackage{authblk}

\usepackage{url}
\usepackage{longtable}
\usepackage[backend=biber,style=authoryear,natbib=true,maxcitenames=2,maxbibnames=99,dashed=false]{biblatex}
\usepackage[colorlinks=true,linkcolor=blue!50!black,citecolor=blue!50!black,urlcolor=blue!50!black]{hyperref}
\usepackage[nameinlink,capitalize,noabbrev]{cleveref}
\usepackage{adjustbox}
\usepackage{listings}
\usepackage{multicol}

\lstdefinestyle{appendixpython2col}{
  language=Python,
  basicstyle=\ttfamily\fontsize{5.2pt}{5.55pt}\selectfont,
  columns=fullflexible,
  keepspaces=true,
  breaklines=true,
  breakatwhitespace=false,
  showstringspaces=false,
  tabsize=2,
  xleftmargin=0pt,
  xrightmargin=0pt,
  aboveskip=0pt,
  belowskip=0pt,
  frame=none
}

\allowdisplaybreaks

\newtheorem{theorem}{Theorem}[section]
\newtheorem{lemma}[theorem]{Lemma}

\theoremstyle{definition}

\newcommand{\E}{\mathbb E}
\newcommand{\ALG}{\mathrm{ALG}}
\newcommand{\OPT}{\mathrm{OPT}}
\newcommand{\LP}{\mathrm{LP}}
\newcommand{\vol}{\operatorname{vol}}
\newcommand{\cost}{\operatorname{cost}}
\newcommand{\mk}{\mathrm{mk}}
\newcommand{\umk}{\mathrm{umk}}
\DeclareMathOperator{\Var}{Var}

\title{Cross-Shift Analysis for Unrelated-Machine Weighted Completion Time:
A \texorpdfstring{$(1.3168+\varepsilon)$}{(1.3168+epsilon)}-Approximation}

\author[1]{Weitian Tong\thanks{Corresponding author:
\href{mailto:wtong.research@gmail.com}{wtong.research@gmail.com}}}

\author[1]{Yao Xu}

\affil[1]{School of Computing, Georgia Southern University,
Statesboro, GA 30461, USA}

\date{}

\begin{document}
\maketitle

\begin{abstract}
We study the problem of minimizing total weighted completion time on
unrelated parallel machines with machine-independent job weights.
The best previous approximation guarantee for this problem is
arbitrarily close to $1.36$, due to Li [SODA 2025], who developed a
configuration-LP and iterative-rounding framework based on randomly
shifted geometric size classes and a computer-assisted final analysis.

We improve the approximation guarantee to arbitrarily close to
$1.3168$. The scheduling algorithm retains Li's
configuration-LP and iterative-rounding framework, with a different fixed
geometric class ratio.  The improvement comes from a new analysis of the
random geometric shift.  For each machine and Smith prefix, we normalize
the physical prefix before averaging over the shift, so that its normalized
size distribution and configuration moments remain fixed across all shifts.
A scale-dependent configuration bound and a cross-shift averaging argument
then reduce the approximation analysis to a one-dimensional certificate.

The final certificate is computer assisted and reproducible.  Its rational
data are verified exactly where possible, while the remaining continuous
inequalities are certified using directed interval arithmetic.  
This yields the first improvement over Li's guarantee for
the machine-independent-weight model.
\end{abstract}


\section{Introduction}\label{sec:intro}

Let $J$ be a set of $n$ jobs and $M$ a set of $m$ unrelated machines.  Job $j\in J$ has a machine-independent weight $w_j>0$ and processing time $p_{ij}>0$ on machine $i\in M$; 
infeasible machine-job pairs can be removed, or equivalently assigned processing time $+\infty$.  
Every job is assigned nonpreemptively to one feasible machine.  
If $C_j$ denotes the completion time of job $j$, the objective is
$\min \sum_{j\in J} w_jC_j$.
In standard three-field notation \citep{graham1979optimization,pinedo2016scheduling}, 
the problem is $R\|\sum_j w_jC_j$.

The model captures the allocation of prioritized jobs to heterogeneous
processing resources: processing times depend on the selected machine,
while completion costs also depend on the sequencing of jobs assigned to
that machine.  For any fixed assignment, Smith's rule orders the jobs on
machine $i$ by nonincreasing $w_j/p_{ij}$ and is optimal
\citep{Smith1956}.  Hence the principal algorithmic difficulty is the
machine assignment, whose choices determine both processing times and the
resulting completion costs.  This interaction has made
$R\|\sum_j w_j C_j$ a longstanding benchmark for mathematical-programming
relaxations, randomized rounding, and approximation algorithms
\citep{VredeveldHurkens2002}.

The problem is strongly NP-hard and APX-hard
\citep{HoogeveenSchuurmanWoeginger2001}.  Classical mathematical-programming
and randomized-rounding approaches established the $3/2$ benchmark
\citep{Chudak1999,SethuramanSquillante1999,Skutella2001,
SchulzSkutella2002}.  \citet{BansalSrinivasanSvensson2021} first broke this
barrier using strong negative-correlation rounding.  Subsequent work
improved the approximation ratio through increasingly refined dependent
rounding and correlated-selection mechanisms
\citep{Li2020,ImShadloo2020,ImLi2023,BavejaQuSrinivasan2024,Harris2024}.

Most recently for the machine-independent-weight model,
\citet{Li2025} developed an iterative-rounding framework that weakens the
correlation requirements needed in earlier approaches.  Within each marked
machine--class group, at most one marked edge is selected, while the
interaction between an unmarked edge and the marked group is controlled
only in aggregate.  Combined with a configuration LP and a
computer-assisted final analysis, this yields a randomized
$(1.36+\varepsilon)$-approximation.

\subsection{Main result and contributions}

Our main result is the following.

\begin{theorem}\label{thm:main}
For every fixed $\varepsilon>0$, there is a randomized polynomial-time
$(1.3168+\varepsilon)$-approximation algorithm for
$R\left\|\sum_j w_jC_j\right.$.
\end{theorem}

To our knowledge, this is the first improvement over Li's
$1.36+\varepsilon$ guarantee \citep{Li2025} for
$R\|\sum_j w_jC_j$ with machine-independent job weights.
Our result retains Li's iterative-rounding mechanism unchanged, apart from using a different
fixed geometric class ratio.  The improvement instead comes from a new
cross-shift analysis: for each physical Smith prefix, we normalize before
averaging over the random geometric shift, thereby preserving a common
normalized size distribution and configuration moments across all
shifts.  An affine lower bound on the class saving then converts the
shifted class structure into a one-dimensional sliding-window transform,
which can be combined with a new scale-dependent configuration bound.

The analysis is based on fixing one physical Smith prefix before averaging
over the random geometric shift.  Let $V$ be its fractional processing
volume.  Normalizing by $V$ gives a random normalized size $\mathsf R$
and two moments $Q$ and $F$ that remain fixed across all shifts.
For every $\gamma\ge0$, we prove the scale-dependent bound
\(
F
\ge
2\gamma-\gamma^2
+
\E\!\left[
\frac{(\mathsf R-\gamma)_+^2}{\mathsf R}
\right].
\)

For Li's class saving, we replace each capped quadratic term by an affine
lower bound and average these bounds over the common logarithmic shift.
This produces the sliding-window function
\(
\mathcal W(r)
=
2\int_{\log_\rho r-1}^{\log_\rho r}\psi(s)\,ds
\)
for a scale profile $\psi$.  Combining the scale-dependent configuration
bound with this cross-shift saving bound reduces the entire prefix
analysis to one pointwise inequality and one scalar condition.
For
$\rho=\frac{53}{25}$,
$\gamma=\frac{89}{100}$,
an explicit rational $232$-cell profile satisfies these conditions and
certifies the factor $1.3168$.

\subsection{Related work and scope}

For restricted variants, stronger approximation schemes are known.
A fixed number of unrelated machines admits PTAS results
\citep{AfratiEtAl1999,AfratiBampisKenyonMilis2000}; identical and uniformly
related machines also admit PTASes for weighted completion time
\citep{SkutellaWoeginger2000,ChekuriKhanna2001}.  Specialized guarantees
are known for uniform Smith-ratio instances
\citep{KalaitzisSvenssonTarnawski2017}.  Our focus is the unrestricted
number of unrelated machines with arbitrary processing times and
machine-independent job weights.

There is also a substantial computational literature on minimizing total
weighted completion time on parallel machines.  
For the same problem
$R\|\sum_j w_jC_j$, \citet{VredeveldHurkens2002} compare approximation
algorithms computationally, while \citet{YuEtAl2026} develop an exact
branch-price-and-cut algorithm.  For identical parallel machines,
\citet{KowalczykLeus2018} propose a branch-and-price approach, and
\citet{LendlPferschyRener2024} obtain an FPTAS for a related rescheduling
variant.  These results complement our focus on improving the
polynomial-time approximation guarantee for
$R\|\sum_j w_jC_j$ with machine-independent job weights.

The more general model in which the job weight may depend on the chosen
machine admits related negative-correlation techniques.  Recently,
\citet{HarrisLiRajuValieva2026} introduced the Dirichlet mechanism and
obtained a $1.387$ approximation for that broader model.  This result and
Theorem~\ref{thm:main} therefore address different problem scopes: the
present analysis relies on the machine-independent sizes obtained after
Li's size-weight interchange.
Very recently, \citet{Harris2026Note} developed another strong
negative-correlation rounding method, combining random-walk and
Dirichlet-based ideas and improving the uniform pairwise
negative-correlation guarantee.  That work concerns the rounding
primitive itself and does not give a new approximation ratio for the
machine-independent weighted-completion-time problem.

\paragraph{Organization.}
\Cref{sec:li} states the precise configuration-LP and rounding facts
inherited from Li.  \Cref{sec:newanalysis} develops the fixed-prefix
normalization, the scale-dependent configuration bound, the cross-shift
saving bound, and the master certificate; it then verifies the explicit
$1.3168$ certificate and derives Theorem~\ref{thm:main}.
\Cref{sec:conclusion} concludes the paper.  
A proof-dependency map is given in Appendix \ref{app:proof-map}; 
the complete rational certificate is provided in Appendix \ref{app:certificate};
and a compact version of the verification program is included in Appendix \ref{app:verifier}.

\section{Li's configuration-LP rounding framework}\label{sec:li}

We use Li's analysis through the prefix identities stated below.
This section recalls the notation and the corresponding results from
\citet{Li2025}.

\subsection{Size-weight swap and configuration LP}

Following \citet{Li2025}, interchange processing times and weights.  For each feasible original pair $(i,j)$, define the transformed processing time $p_j:=w_j$ and transformed weight $w_{ij}:=p_{ij}$.  Li proves that every assignment has the same Smith-order cost before and after this swap \citep[Lemma~2.1]{Li2025}.  From this point through \Cref{sec:newanalysis}, $p_j$ and $w_{ij}$ denote the transformed quantities.  Thus $p_j>0$ is machine independent, $w_{ij}>0$ may depend on $i$, and the Smith ratio on machine $i$ is
\(
\sigma_{ij}:=\frac{w_{ij}}{p_j}.
\)

For each machine $i$, fix once and for all an arbitrary total order of
the jobs that refines nonincreasing $\sigma_{ij}$; ties in the Smith
ratio are broken according to this fixed order.  We use the same
tie-breaking rule throughout the marking procedure, the definition of
Smith prefixes, and the final sequencing of jobs.

For each machine $i\in M$, let
$J_i:=\{j\in J:(i,j)\text{ is feasible}\}$,
$\mathcal F_i:=\{f:f\subseteq J_i\}$.
For $f\in\mathcal F_i$, let $\cost_i(f)$ be the weighted completion time
of the jobs in $f$ when processed on machine $i$ in Smith order.  Li's
configuration LP is
\begin{align}
\min\quad
&\sum_{i\in M}\sum_{f\in\mathcal F_i}y_{if}\cost_i(f),
\label{lp:obj}\\
\text{s.t.}\quad
&\sum_{f\in\mathcal F_i}y_{if}=1
&& i\in M,
\label{lp:machine}\\
&\sum_{i\in M}
  \sum_{\substack{f\in\mathcal F_i\\j\in f}}y_{if}=1
&& j\in J,
\label{lp:job}\\
&y_{if}\ge0
&& i\in M,\ f\in\mathcal F_i.
\label{lp:nonneg}
\end{align}

For every machine-job pair, define
$z_{ij}:= \sum_{f\in\mathcal F_i, j\in f}y_{if}$ if $j\in J_i$; and 0 if $j\notin J_i$.
Then $\sum_{i\in M}z_{ij}=1$ for every $j\in J$.
Whenever a machine $i$ is fixed below, we write $\sum_f$ as shorthand
for $\sum_{f\in\mathcal F_i}$.

For any job set $S\subseteq J$, write
$p(S):=\sum_{j\in S}p_j$,
and, after a machine $i$ is fixed, define its fractional processing volume on $i$ by
$\vol_i(S):=\sum_{j\in S}z_{ij}p_j.$
These two job-set quantities will be used throughout the analysis.

\subsection{Random geometric classes and Li's edge notation}

Fix a constant class ratio $\rho>1$.  Draw $\beta\in[1,\rho)$ so that $\ln\beta$ is uniformly distributed on $[0,\ln\rho)$.  For each $k\in\mathbb Z$, Li defines
\[
J_k:=\{j\in J:\beta\rho^k\le p_j<\beta\rho^{k+1}\}.
\]
Since the transformed processing times $p_j$ are machine independent, the sets $(J_k)_{k\in\mathbb Z}$ form a global partition of the jobs.

Li constructs a bipartite support multigraph between the machines $M$ and jobs $J$.  Let $E$ be its edge multiset and let $x_e\in(0,1]$ be the fractional value associated with edge $e\in E$.  If $e$ is incident to job $j$, set $p_e:=p_j$ and define
$\vol(e):=x_ep_e$,
$\vol(E'):=\sum_{e\in E'}\vol(e)$
for 
$E'\subseteq E$.
For a machine $i$, let $\delta_i\subseteq E$ denote its incident edges, and let $\delta_i^k\subseteq\delta_i$ denote the edges between $i$ and jobs in $J_k$.  Let $\delta_i^{k,\mk}$ and $\delta_i^{k,\umk}$ be the marked and unmarked edges in $\delta_i^k$, respectively.

Within each $\delta_i^k$, the edges are ordered by nonincreasing Smith ratio $w_{ij}/p_j$.  Li marks the first $\beta\rho^k$ units of edge volume, or all of $\delta_i^k$ when its total volume is smaller.  If necessary, one boundary edge is split into a marked and an unmarked parallel edge.  Hence
\(
\vol(\delta_i^{k,\mk})
=
\min\{\vol(\delta_i^k),\beta\rho^k\}
\) holds
\citep[Claim~2.3]{Li2025}.

Li rounds the classes independently using cycle and pseudo-marked-path updates.  The resulting assignment selects at most one marked edge from each set $\delta_i^{k,\mk}$ \citep[Lemmas~2.12-2.13]{Li2025}.  Moreover, correlations between an unmarked edge and the marked edges of the same machine and class are controlled in aggregate rather than edge by edge \citep[Lemma~2.10 and Corollary~2.11]{Li2025}.  These properties are the source of the negative correction in Li's expected-cost bound stated below.

\subsection{The inherited prefix formulas}

Fix a machine $i\in M$.  For each job $j$, let
\(
\sigma_j:=\frac{w_{ij}}{p_j}
\)
be its Smith ratio on machine $i$. 
Index the jobs according to the
fixed Smith order defined above, so that
$\sigma_1\ge\sigma_2\ge\cdots\ge\sigma_n$,
$\sigma_{n+1}:=0$.
For $j^*\in\{1,\ldots,n\}$, let
$[j^*]:=\{1,\ldots,j^*\}$.
Let $\ALG$ denote the random integral assignment produced by Li's iterative
rounding, and let $\E_{\rm rnd}[\cdot\mid\beta]$ denote expectation over the
rounding randomness conditional on the class shift $\beta$.

The following theorem collects the results from Li's analysis that are used in this paper.

\begin{theorem}[Li's prefix interface]
\label{thm:li-prefix}
For every machine $i$, its contribution to the configuration LP is
\begin{align}
\sum_f y_{if}\cost_i(f)
=
\sum_{j^*=1}^n
(\sigma_{j^*}-\sigma_{j^*+1})
\frac12
\left(
\sum_{j\in[j^*]}z_{ij}p_j^2
+
\sum_f y_{if}p(f\cap[j^*])^2
\right).
\label{eq:li-lp-prefix}
\end{align}
Conditional on a fixed class shift $\beta$, Li's iterative rounding satisfies
\begin{align}
\E_{\rm rnd}[\cost_i(\ALG)\mid\beta]
\le
\sum_{j^*=1}^n
(\sigma_{j^*}-\sigma_{j^*+1})
\Bigg[
&\sum_{j\in[j^*]}z_{ij}p_j^2
+\frac12\vol_i([j^*])^2
\nonumber\\
&-\frac12\sum_{k\in\mathbb Z}
\min\{\vol_i([j^*]\cap J_k),\beta\rho^k\}^2
\Bigg].
\label{eq:li-round-prefix}
\end{align}
Moreover, the rounding terminates in polynomial time and returns an integral assignment.
\end{theorem}

Equation~\eqref{eq:li-lp-prefix} is Li's Lemma~3.1, and
\eqref{eq:li-round-prefix} is Li's Lemma~3.2.  The correlation, integrality,
and termination properties needed for the latter are established in
Section~2 of \citet{Li2025}.  These statements form the complete interface
with Li's analysis used in the remainder of this paper.

The last term in \eqref{eq:li-round-prefix},
\(
\frac12\sum_{k\in\mathbb Z}
\min\{\vol_i([j^*]\cap J_k),\beta\rho^k\}^2,
\)
is the \emph{class saving}.  For each class $J_k$, the quantity
\(
\min\{\vol_i([j^*]\cap J_k),\beta\rho^k\}
\)
is the marked volume contained in the Smith prefix $[j^*]$, by Li's marking rule.  The square of this marked volume quantifies the reduction in the expected pairwise completion-time contribution created by the negative dependence among marked edges.  Summing over all classes gives the total saving subtracted from the quadratic prefix bound.

Because $\sigma_{j^*}-\sigma_{j^*+1}\ge0$
for every prefix, it is sufficient to compare the two bracketed prefix
expressions in Theorem~\ref{thm:li-prefix} separately for every machine and
Smith prefix.  The next section proves such a comparison with factor
$1.3168$.  After the prefix analysis is complete, we return in
\Cref{sec:global} to sum these inequalities and incorporate the
$(1+\eta)$-approximate solution of the configuration LP.

\section{Cross-shift analysis}
\label{sec:newanalysis}

We now present the new analysis that improves the approximation guarantee
obtained from Li's rounding framework.  All results summarized in
Theorem~\ref{thm:li-prefix} are used as established facts.  We use Li's
rounding procedure and its structural guarantees without modification;
the only algorithmic parameter we change is the fixed geometric class
ratio, which we set to $\rho=\frac{53}{25}$
instead of the value $\rho=2$ used in Li's final approximation analysis.
Our main contribution is a sharper analysis of this same rounding
framework.

The key difference is the treatment of the random geometric shift.  We
first fix a machine and a Smith prefix and normalize this same prefix
before averaging over the shift.  Consequently, all quantities determined
by the fixed prefix remain unchanged as the geometric class boundaries
move.  This common normalization allows the configuration-LP contribution
and Li's class saving to be analyzed separately while preserving their
connection to the same underlying prefix.

The analysis proceeds in five steps.  First, \Cref{sec:normalize}
normalizes a fixed prefix and reduces the desired approximation guarantee
to a single target inequality.  Next, \Cref{sec:variance} derives a
stronger lower bound on the configuration-LP contribution by exploiting
the distribution of job sizes within the prefix.  Then,
\Cref{sec:saving} develops a lower bound on Li's class saving that can be
averaged over all geometric shifts.  \Cref{sec:master} combines these two
bounds and reduces the approximation analysis to a one-variable
pointwise inequality together with a scalar condition.  Finally,
\Cref{sec:certificate} gives a rigorous computer-assisted verification
of an explicit certificate that establishes the approximation factor
$1.3168$.

\subsection{Fixed-prefix normalization and the target inequality}
\label{sec:normalize}

Fix a machine $i$ and a Smith prefix $[j^*]$.  From this point until the global approximation argument, the machine and prefix remain fixed.  Let
\(
V:=\vol_i([j^*])=\sum_{j\in[j^*]}z_{ij}p_j
\)
be their fractional processing volume.  If $V=0$, then every term associated with this prefix in both inherited prefix formulas is zero, so there is nothing to prove.  We therefore assume $V>0$.

\subsubsection{Normalized sizes and shift-independent moments}

For every $j\in[j^*]$, define its normalized processing time by
\(
r_j:=\frac{p_j}{V}.
\)
The quantity
\[
\pi_j:=z_{ij}r_j=\frac{z_{ij}p_j}{V},
\qquad j\in[j^*],
\]
is the fraction of the prefix's total fractional processing volume
contributed by job $j$.  Since
$\sum_{j\in[j^*]}\pi_j=1,$
the values $(\pi_j)_{j\in[j^*]}$ form a probability distribution.

Define a random variable $\mathsf R$ by
$\Pr(\mathsf R=r_j)=\pi_j$,
$j\in[j^*]$.
This distribution is defined from the physical
prefix before the random geometric shift is sampled and therefore does not
depend on $\beta$.  With this choice, the normalized jobwise quadratic term is the first
moment of $\mathsf R$:
\begin{align*}
Q
&:=\frac1{V^2}\sum_{j\in[j^*]}z_{ij}p_j^2
 =\sum_{j\in[j^*]}z_{ij}r_j^2
 =\sum_{j\in[j^*]}\pi_j r_j
 =\E[\mathsf R].
\end{align*}

For each configuration $f\subseteq J$, let
\[
T_f:=\sum_{j\in f\cap[j^*]} r_j
    =\frac{p(f\cap[j^*])}{V}
\]
denote the total processing load contributed by the prefix jobs in
configuration $f$, normalized by the fractional prefix volume $V$.
Since the configuration variables satisfy $\sum_f y_{if}=1$, we may
view $y_{if}$ as a probability distribution over the configurations of
machine $i$.  We then define
\[
F
:=
\frac{1}{V^2}\sum_f y_{if}p(f\cap[j^*])^2
=
\sum_f y_{if}T_f^2.
\]
Thus $F$ is the second moment of the normalized prefix load $T_f$ under
the configuration distribution $(y_{if})_f$.  Its first moment is
\begin{align*}
\sum_fy_{if}T_f
&=
\frac1V
\sum_{j\in[j^*]}p_j\sum_{f\ni j}y_{if} 
=
\frac1V
\sum_{j\in[j^*]}z_{ij}p_j
=1.
\end{align*}
Consequently,
\(
F=\E_f[T_f^2]
 =1+\Var_f(T_f)
 \ge1.
\)

A second elementary bound will be useful for comparison later:
\begin{align*}
F
&=\sum_fy_{if}T_f^2 
\ge \sum_fy_{if}\sum_{j\in f\cap[j^*]}r_j^2 
=\sum_{j\in[j^*]}r_j^2\sum_{f\ni j}y_{if}
 =\sum_{j\in[j^*]}z_{ij}r_j^2
 =Q.
\end{align*}

We next strengthen $F\ge Q$ using the locations of the normalized sizes.

\subsubsection{Normalized class saving and the target inequality}

For a fixed shift $\beta$, the normalized lower boundary of class $k$ is
\(
\tau_k:=\frac{\beta\rho^k}{V},
\)
and the normalized fractional volume of the prefix in that class is
\(
v_k:=\frac{\vol_i([j^*]\cap J_k)}{V}.
\)
By the definition of $\mathsf R$ and of the job class $J_k$,
\(
v_k=\Pr\left(\tau_k\le \mathsf R<\rho\tau_k \right).
\)
Thus $0\le v_k\le1$ and $\sum_kv_k=1$.

For the fixed prefix and shift $\beta$, the class saving in \eqref{eq:li-round-prefix} can be written as
\(
\frac12\sum_{k\in\mathbb Z}
\min\{\vol_i([j^*]\cap J_k),\beta\rho^k\}^2
=
\frac{V^2}{2}
\sum_{k\in\mathbb Z}
\min\{v_k,\tau_k\}^2.
\)
Define the normalized class saving by
\[
S_\beta
:=
\sum_{k\in\mathbb Z}\min\{v_k,\tau_k\}^2,
\qquad
\overline S
:=
\E_\beta[S_\beta].
\]
Thus, $V^2S_\beta/2$ is the class saving for this prefix, while
$\overline S$ is its normalized expectation over the random shift.

Dividing the algorithmic prefix bracket in Theorem~\ref{thm:li-prefix}
by $V^2$ and averaging over $\beta$ gives
\(
Q+\frac12-\frac12\overline S,
\)
whereas the corresponding normalized configuration-LP bracket is
\(
\frac12(Q+F).
\)
Therefore, to obtain an approximation factor $\alpha>1$ for this prefix,
it suffices to prove
\begin{equation}\label{eq:target}
2Q+1-\overline S
\le
\alpha(Q+F).
\end{equation}

This is the main inequality analyzed in the remainder of the paper.
The quantities $Q$, $F$, and the distribution of $\mathsf R$ depend only
on the fixed machine and prefix and are therefore independent of the
random shift $\beta$; only the class-saving term $S_\beta$ varies with
the shift.

\subsection{A scale-dependent lower bound on the configuration moment}
\label{sec:variance}

The first resource in the target inequality \eqref{eq:target}, is the configuration second moment $F$.  The elementary inequality $F\ge Q$ uses only the individual squares inside each configuration.  We can obtain a stronger family of bounds by comparing every configuration load with a common level $\gamma\ge0$.

For $x\in\mathbb R$, write $x_+:=\max\{x,0\}$.  For a fixed comparison level $\gamma\ge0$, define
\[
  h_\gamma(r):=\frac{(r-\gamma)_+^2}{r},\qquad r>0.
\]
The function $h_\gamma$ measures the squared excess of a normalized size above $\gamma$, divided by the size itself.  The division by $r$ is tailored to the volume-biased distribution: since job $j$ is sampled with probability $\pi_j=z_{ij}r_j$, it cancels that extra factor of $r_j$.  Explicitly,
\[
\E[h_\gamma(\mathsf R)]
=\sum_{j\in[j^*]}\pi_j\frac{(r_j-\gamma)_+^2}{r_j} 
=\sum_{j\in[j^*]}z_{ij}(r_j-\gamma)_+^2
 =\sum_{\substack{j\in[j^*]\\ r_j\ge\gamma}}z_{ij}(r_j-\gamma)^2.
\]

\begin{theorem}[Scale-dependent configuration bound]
\label{thm:variance}
For every $\gamma\ge0$,
\[
F\ge 2\gamma-\gamma^2+\E[h_\gamma(\mathsf R)].
\]
\end{theorem}

\begin{proof}
Fix a configuration $f$ and let
\(
H_f:=\{j\in f\cap[j^*]:r_j\ge\gamma\}.
\)
If $H_f=\varnothing$, then $(T_f-\gamma)^2\ge0$ and there is nothing to show for this configuration.  Otherwise $|H_f|\ge1$.  Separating jobs above and below $\gamma$ gives
\begin{align*}
T_f-\gamma
&=\sum_{j\in H_f}r_j
  +\sum_{\substack{j\in f\cap[j^*]\\ r_j<\gamma}}r_j-\gamma 
=\sum_{j\in H_f}(r_j-\gamma)
  +( |H_f|-1)\gamma
  +\sum_{\substack{j\in f\cap[j^*]\\ r_j<\gamma}}r_j 
\ge \sum_{j\in H_f}(r_j-\gamma)\ge0.
\end{align*}
Every term discarded in the last line is nonnegative.  Squaring and then using $(\sum_q x_q)^2\ge\sum_qx_q^2$ for nonnegative $x_q$ yields
\(
(T_f-\gamma)^2
\ge\sum_{j\in H_f}(r_j-\gamma)^2.
\)
Multiply by $y_{if}$ and sum over all configurations.  The two sides simplify as follows:
\begin{align*}
\sum_fy_{if}(T_f-\gamma)^2
&=F-2\gamma\sum_fy_{if}T_f+\gamma^2\sum_fy_{if}
 =F-2\gamma+\gamma^2, \\
\sum_fy_{if}\sum_{j\in H_f}(r_j-\gamma)^2
&=\sum_{\substack{j\in[j^*]\\ r_j\ge\gamma}}
(r_j-\gamma)^2\sum_{f\ni j}y_{if} 
=\sum_{\substack{j\in[j^*]\\ r_j\ge\gamma}}z_{ij}(r_j-\gamma)^2
 =\E[h_\gamma(\mathsf R)].
\end{align*}
For the first identity we used $\sum_fy_{if}T_f=1$ and $\sum_fy_{if}=1$.  For the second, we exchanged two finite sums and then used $z_{ij}=\sum_{f\ni j}y_{if}$.  Consequently
\(
F-2\gamma+\gamma^2\ge\E[h_\gamma(\mathsf R)],
\)
which proves the theorem.
\end{proof}

The parameter $\gamma$ determines the size threshold at which the strengthened lower bound on $F$ places additional emphasis.  Two endpoint choices recover familiar bounds.  When $\gamma=0$, we have
\(
h_0(r)=r,
\)
and the theorem gives
\(
F\ge \E[\mathsf R]=Q,
\)
which is exactly the elementary bound established above.  When $\gamma=1$, the theorem gives
\(
F\ge 1+\E[h_1(\mathsf R)],
\)
so, in addition to the baseline bound $F\ge1$, it captures an extra contribution from normalized job sizes exceeding $1$.

More generally, choosing $\gamma$ between $0$ and $1$ allows the lower bound on $F$ to focus on a particular range of normalized job sizes.  This flexibility is useful because the cross-shift saving bound is not equally strong for all values of $\mathsf R$.  We therefore choose $\gamma$ jointly with the cross-shift certificate so that the strengthened configuration bound compensates for the size range where the saving estimate is weakest.  The explicit certificate used later takes
\(
\gamma=\frac{89}{100}.
\)

\subsection{A cross-shift lower bound on the class saving}
\label{sec:saving}

The second quantity that can help satisfy the target inequality
\eqref{eq:target} is the average class saving $\overline S$.  For a fixed
shift $\beta$, recall that
\(
S_\beta
=
\sum_{k\in\mathbb Z}\min\{v_k,\tau_k\}^2.
\) and 
\(
\overline S
:=
\E_\beta[S_\beta]\).
Thus, the contribution of a single class with normalized volume $v$ and
normalized lower boundary $\tau$ is
\(
\min\{v,\tau\}^2.
\)
This expression is nonlinear in $v$, which makes it difficult to average
directly over the random shift because both the class boundaries and the
corresponding class volumes change with $\beta$.

Our first step is therefore to lower-bound $\min\{v,\tau\}^2$ by an
expression that is affine in $v$.  Such a bound has the form
\(
2\lambda v-\lambda^2,
\)
where $\lambda$ will later be chosen as a function of the class scale.
The linear dependence on $v$ is crucial: after summing over classes and
averaging over the random shift, the terms involving $v$ can be
reorganized according to the normalized job sizes.  This is the
mechanism that will convert the collection of randomly shifted classes
into a one-dimensional bound.

\subsubsection{An affine lower bound for one class}

Fix a normalized class boundary $\tau>0$ and a normalized class volume
$v\in[0,1]$.  The saving contributed by this class is
\(
\min\{v,\tau\}^2.
\)
To average this quantity over the random class shift, it is useful to
replace it by a lower bound that is affine in $v$.

For a parameter $\lambda\ge0$, consider
\(
2\lambda v-\lambda^2.
\)
This is the tangent line to $v^2$ at $v=\lambda$, and
\(
v^2-(2\lambda v-\lambda^2)=(v-\lambda)^2\ge0.
\)
Hence it is always a lower bound for $v^2$.

When $0<\tau<1$, however, the class-saving function is capped:
\[
\min\{v,\tau\}^2
=
\begin{cases}
v^2, & 0\le v\le\tau,\\
\tau^2, & \tau\le v\le1.
\end{cases}
\]
Thus, the tangent line must also remain below the plateau $\tau^2$.
Since $2\lambda v-\lambda^2$ is nondecreasing in $v$ for
$\lambda\ge0$, it is enough to require
\(
2\lambda-\lambda^2\le\tau^2.
\)
For $0<\tau<1$, the largest $\lambda\in[0,1]$ satisfying this condition is
\(
1-\sqrt{1-\tau^2}.
\)
Accordingly, define
\[
b(\tau):=
\begin{cases}
1-\sqrt{1-\tau^2}, & 0<\tau<1,\\[2mm]
1, & \tau\ge1.
\end{cases}
\]
For $0<\tau<1$, $b(\tau)$ is the largest admissible tangent parameter
that keeps the affine bound below the plateau.  When $\tau\ge1$, we
have $\min\{v,\tau\}^2=v^2$ for all $v\in[0,1]$, so no additional
restriction is needed beyond $\lambda\le1$.
Figure~\ref{fig:affine-class-bound} illustrates the nontrivial case
$0<\tau<1$.  

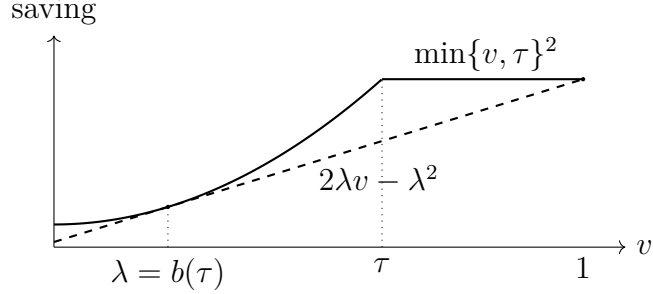
\begin{figure}[htbp]
\centering
\begin{tikzpicture}[x=7cm,y=5cm]
    \def\taudraw{0.62}
    \pgfmathsetmacro{\lamdraw}{1-sqrt(1-\taudraw*\taudraw)}
    \pgfmathsetmacro{\plateau}{\taudraw*\taudraw}

    \draw[->] (0,-0.06) -- (1.08,-0.06) node[right] {$v$};
    \draw[->] (0,-0.06) -- (0,0.50) node[above] {saving};

    \draw[thick,domain=0:\taudraw,samples=100]
        plot (\x,{\x*\x});
    \draw[thick]
        (\taudraw,\plateau) -- (1,\plateau);

    \draw[dashed,thick,domain=0:1,samples=2]
        plot (\x,{2*\lamdraw*\x-\lamdraw*\lamdraw});

    \draw[dotted]
        (\taudraw,-0.06) -- (\taudraw,\plateau);
    \draw[dotted]
        (\lamdraw,-0.06) -- (\lamdraw,{\lamdraw*\lamdraw});

    \fill (\lamdraw,{\lamdraw*\lamdraw}) circle (0.8pt);
    \fill (1,\plateau) circle (0.8pt);

    \node[below] at (\lamdraw,-0.06) {$\lambda=b(\tau)$};
    \node[below] at (\taudraw,-0.06) {$\tau$};
    \node[below] at (1,-0.06) {$1$};

    \node[anchor=west] at (0.66,0.45)
        {$\min\{v,\tau\}^2$};
    \node[anchor=west] at (0.48,0.12)
        {$2\lambda v-\lambda^2$};

\end{tikzpicture}
\caption{Affine lower bound for a class with $0<\tau<1$.
The solid curve is $\min\{v,\tau\}^2$.  For the maximal admissible
choice $\lambda=b(\tau)$, the dashed line is tangent to $v^2$ at
$v=\lambda$ and reaches the plateau $\tau^2$ at $v=1$.}
\label{fig:affine-class-bound}
\end{figure}

\begin{lemma}\label{lem:tangent}
For every $\tau>0$, $v\in[0,1]$, and
$0\le\lambda\le b(\tau)$, we have
\(
\min\{v,\tau\}^2
\ge
2\lambda v-\lambda^2.
\)
\end{lemma}

\begin{proof}
If $v\le\tau$, then $\min\{v,\tau\}^2=v^2$, and
\(
v^2-(2\lambda v-\lambda^2)
=
(v-\lambda)^2
\ge0.
\)
Suppose now that $v\ge\tau$.  This case requires attention only when
$\tau<1$, for which $\min\{v,\tau\}^2=\tau^2$.  Since
$v\le1$ and $\lambda\ge0$,
\(
2\lambda v-\lambda^2
\le
2\lambda-\lambda^2
\le
\tau^2,
\)
where the last inequality follows from $\lambda\le b(\tau)$.
Therefore,
\(
2\lambda v-\lambda^2
\le
\min\{v,\tau\}^2,
\)
which completes the proof.
\end{proof}

\subsubsection{A scale profile and the unit-window function}

Lemma~\ref{lem:tangent} allows the affine parameter $\lambda$ to depend
on the class threshold.  Since the class thresholds are geometric, it is
convenient to parameterize them on the logarithmic scale.  Write
\(
s:=\log_\rho \tau\),
so that
\(\tau=\rho^s.
\)

A \emph{scale profile} is a measurable function
$\psi:\mathbb R\to\mathbb R_{\ge0}$ satisfying
\[
0\le \psi(s)\le b(\rho^s),~ s\in\mathbb R,
\text{ and }
\Pi:=\int_{-\infty}^{\infty}\psi(s)^2\,ds<\infty.
\]
At logarithmic scale $s$, we use
\(
\lambda=\psi(s)
\)
in Lemma~\ref{lem:tangent}.  The cap condition guarantees that this is
a valid choice for the class threshold $\tau=\rho^s$.
The quantity $\Pi$ records the total cost associated with these affine
lower bounds.  Indeed, each application of Lemma~\ref{lem:tangent}
contributes a subtractive term $\lambda^2$, and averaging these terms
over all shifts will produce the integral $\Pi$.  The explicit profile
used in the final certificate is particularly simple: it is piecewise
constant and has finite support.

We next describe how a fixed normalized job size interacts with the
profile.  Let $r>0$ and set
\(
u:=\log_\rho r.
\)
A geometric class whose lower boundary is $\rho^s$ contains $r$ exactly
when
\(
\rho^s\le r<\rho^{s+1},
\)
or equivalently,
\(
u-1<s\le u.
\)
Thus, over one complete shift of the geometric grid, a job of normalized
size $r$ encounters precisely the profile values on the unit interval
$(u-1,u]$.  This motivates the \emph{unit-window function}
\[
\mathcal W(r)
:=
2\int_{\log_\rho r-1}^{\log_\rho r}\psi(s)\,ds,
\qquad r>0.
\]
The factor $2$ comes from the linear term $2\lambda v$ in
Lemma~\ref{lem:tangent}.  Hence $\mathcal W(r)$ is the total linear
coefficient associated with a normalized size $r$ after averaging over
one full logarithmic shift.

It is useful to express the normalized geometric grid directly in
logarithmic coordinates.  Write
\(
\log_\rho(\beta/V)=z+\theta\),
\(
z\in\mathbb Z,\quad \theta\in[0,1).
\)
The integer $z$ only reindexes the classes.  Thus, after reindexing, the
normalized class boundaries are
\(
\tau_k=\rho^{k+\theta}\),
\(k\in\mathbb Z.
\)
Equivalently, in the logarithmic coordinate
\(
u:=\log_\rho r,
\)
the classes are the unit intervals
\(
[k+\theta,k+\theta+1).
\)

Because the physical prefix is normalized before the random shift is
sampled, the normalized sizes and configuration moments remain fixed while
the logarithmic class grid translates across them.
\Cref{fig:cross-shift} illustrates this cross-shift viewpoint and the origin
of the unit-window transform used below.

\begin{figure}[htbp]
\centering
\resizebox{0.98\linewidth}{!}{%
\begin{tikzpicture}[
    x=0.92cm,
    y=0.90cm,
    >=Latex,
    job/.style={circle,fill=black,inner sep=1.45pt},
    boundary/.style={densely dashed,thin},
    guide/.style={densely dotted,black!35}
]

\node[font=\bfseries] at (4.25,5.55)
    {(a) One fixed prefix under shifted class grids};

\draw[->,thick]
    (0,4.60) -- (8.65,4.60)
    node[right] {$u=\log_\rho r$};

\def\xone{0.95}
\def\xtwo{2.05}
\def\xthree{3.35}
\def\xfour{4.55}
\def\xfive{6.00}
\def\xsix{7.15}

\foreach \x/\lab in {
    \xone/1,
    \xtwo/2,
    \xthree/3,
    \xfour/4,
    \xfive/5,
    \xsix/6
}{
    \draw[guide]
        (\x,1.15) -- (\x,4.45);

    \node[job]
        at (\x,4.60) {};

    \node[above=2pt,font=\scriptsize]
        at (\x,4.60)
        {$u_{\lab}$};
}

\foreach \y/\off/\phase in {
    3.45/0.22/1,
    2.45/0.55/2,
    1.45/0.86/3
}{
    \draw[thin]
        (0,\y) -- (8.45,\y);

    \node[anchor=east,font=\small]
        at (-0.18,\y)
        {$\theta_{\phase}$};

    \foreach \n in {0,...,8}{
        \pgfmathsetmacro{\bx}{\off+\n}
        \draw[boundary]
            (\bx,\y-0.28) -- (\bx,\y+0.28);
    }

    \foreach \x in {
        \xone,\xtwo,\xthree,\xfour,\xfive,\xsix
    }{
        \node[job] at (\x,\y) {};
    }
}

\draw[<->,thin]
    (0.22,3.86) -- (1.22,3.86)
    node[midway,above=1pt,font=\scriptsize]
    {$1$};

\node[anchor=west,font=\scriptsize]
    at (1.38,3.86)
    {class width in $u$-space};

%
\begin{scope}[xshift=2cm]

\node[font=\bfseries] at (12.20,5.55)
    {(b) Averaging for one fixed normalized size};

\draw[->,thick]
    (9.35,4.20) -- (15.40,4.20)
    node[right] {$u$};

\coordinate (S)  at (10.55,4.20);
\coordinate (U)  at (12.95,4.20);
\coordinate (SP) at (13.35,4.20);

\draw[very thick]
    (S) -- (SP);

\draw[thin]
    (10.55,3.98) -- (10.55,4.42);

\draw[thin]
    (13.35,3.98) -- (13.35,4.42);

\node[below=3pt,font=\scriptsize]
    at (S)
    {$s=k+\theta$};

\node[below=3pt,font=\scriptsize]
    at (SP)
    {$s+1$};

\node[job]
    at (U) {};

\node[above=3pt,font=\scriptsize]
    at (U)
    {$u_j$};

\node[above=9pt,font=\scriptsize]
    at (11.95,4.20)
    {$s\le u_j<s+1$};

\draw[->,thick]
    (9.45,2.15) -- (15.10,2.15);

\coordinate (UM) at (10.15,2.15);
\coordinate (UJ) at (12.95,2.15);

\draw[very thick]
    (UM) -- (UJ);

\draw[thin]
    (10.15,1.93) -- (10.15,2.37);

\draw[thin]
    (12.95,1.93) -- (12.95,2.37);

\node[below=3pt,font=\scriptsize]
    at (UM)
    {$u_j-1$};

\node[below=3pt,font=\scriptsize]
    at (UJ)
    {$u_j$};

\draw[->,thin]
    (10.45,2.62) -- (12.65,2.62)
    node[midway,above=1pt,font=\scriptsize]
    {$s=k+\theta$ as $\theta$ ranges over $[0,1)$};


\end{scope}

\end{tikzpicture}%
}

\caption{
Cross-shift averaging in logarithmic coordinates.
(a) The normalized job sizes $u_j=\log_\rho r_j$ are fixed, while the class
boundaries $k+\theta$ translate with the random phase.
(b) For a fixed $u_j$, the lower boundary $s=k+\theta$ of its containing
class ranges over $u_j-1<s\le u_j$, giving the window
$\mathcal W(r_j)=2\int_{u_j-1}^{u_j}\psi(s)\,ds$.
The job locations are schematic.
}
\label{fig:cross-shift}
\end{figure}
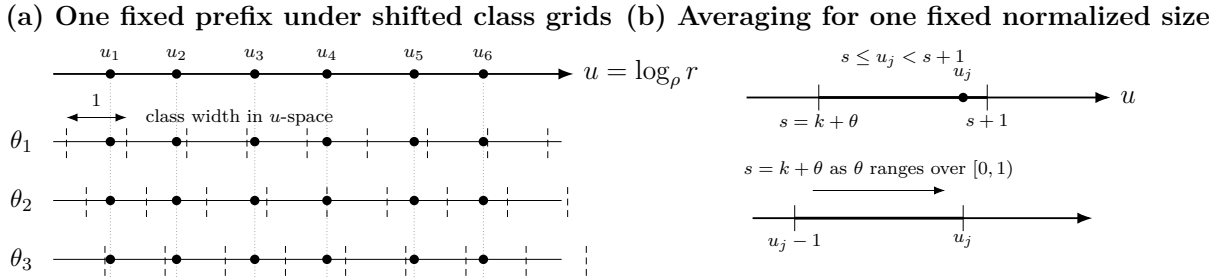

\begin{lemma}[Cross-shift saving bound]
\label{lem:saving}
For the fixed prefix and every scale profile $\psi$ defined above, we have
\(
\overline S
\ge
\E[\mathcal W(\mathsf R)]-\Pi.
\)
\end{lemma}

\begin{proof}
Recall the logarithmic phase $\theta\in[0,1)$ defined above.  After the
integer reindexing of the classes, the normalized class boundaries are
\(
\tau_k=\rho^{k+\theta}\),
\(k\in\mathbb Z,
\)
and $\theta$ is uniform on $[0,1)$.

For this shifted grid, define
\(
v_k(\theta)
:=
\Pr\!\left(
\rho^{k+\theta}\le \mathsf R<\rho^{k+\theta+1}
\right).
\)
Then
\(
S_\beta
=
\sum_{k\in\mathbb Z}
\min\{v_k(\theta),\rho^{k+\theta}\}^2.
\)
At the threshold $\rho^{k+\theta}$, choose the affine parameter
\(
\lambda=\psi(k+\theta).
\)
The definition of a scale profile gives
\(
0\le\psi(k+\theta)\le b(\rho^{k+\theta}),
\)
so Lemma~\ref{lem:tangent} applies to every class.  Summing its bound
over $k$ yields
\begin{align*}
S_\beta
&\ge
\sum_{k\in\mathbb Z}
\left(
2\psi(k+\theta)v_k(\theta)
-
\psi(k+\theta)^2
\right) 
=
2\sum_{k\in\mathbb Z}
\psi(k+\theta)v_k(\theta)
-
\sum_{k\in\mathbb Z}
\psi(k+\theta)^2.
\end{align*}

We now average the two terms over the uniform phase
$\theta\in[0,1)$.
%
For the second term,
\begin{align*}
\int_0^1
\sum_{k\in\mathbb Z}\psi(k+\theta)^2\,d\theta
&=
\sum_{k\in\mathbb Z}
\int_0^1\psi(k+\theta)^2\,d\theta 
=
\sum_{k\in\mathbb Z}
\int_k^{k+1}\psi(s)^2\,ds 
=
\int_{-\infty}^{\infty}\psi(s)^2\,ds 
=
\Pi.
\end{align*}

It remains to average the linear term.  By the definition of
$\mathsf R$,
\[
v_k(\theta)
=
\sum_{j\in[j^*]}
\pi_j\,
\mathbf 1
\left\{
\rho^{k+\theta}
\le r_j
<
\rho^{k+\theta+1}
\right\}.
\]
Therefore,
\begin{align*}
&\int_0^1
2\sum_{k\in\mathbb Z}
\psi(k+\theta)v_k(\theta)\,d\theta 
=
\sum_{j\in[j^*]}\pi_j
\int_0^1
2\sum_{k\in\mathbb Z}
\psi(k+\theta)
\mathbf 1
\left\{
\rho^{k+\theta}
\le r_j
<
\rho^{k+\theta+1}
\right\}
\,d\theta.
\end{align*}

Fix a job $j$ and write
\(
u_j:=\log_\rho r_j.
\)
The indicator is equal to one exactly when
\(
k+\theta\le u_j<k+\theta+1,
\)
or, with $s:=k+\theta$,
\(
u_j-1<s\le u_j.
\)
As $\theta$ ranges over $[0,1)$, the corresponding values of $s$ in this
interval are encountered exactly once.  Hence
\begin{align*}
&\int_0^1
2\sum_{k\in\mathbb Z}
\psi(k+\theta)
\mathbf 1
\left\{
\rho^{k+\theta}
\le r_j
<
\rho^{k+\theta+1}
\right\}
\,d\theta 
=
2\int_{u_j-1}^{u_j}\psi(s)\,ds
=
\mathcal W(r_j).
\end{align*}
Substituting this identity gives
\begin{align*}
\int_0^1
2\sum_{k\in\mathbb Z}
\psi(k+\theta)v_k(\theta)\,d\theta
&=
\sum_{j\in[j^*]}
\pi_j\mathcal W(r_j) 
=
\E[\mathcal W(\mathsf R)].
\end{align*}

Combining the averaged linear term with the profile cost yields
\(
\overline S
\ge
\E[\mathcal W(\mathsf R)]-\Pi,
\)
as claimed.
\end{proof}

\subsection{Master certificate}
\label{sec:master}

We now combine the two lower bounds developed above.  By
Theorem~\ref{thm:variance}, the configuration moment satisfies
\(
F\ge 2\gamma-\gamma^2+\E[h_\gamma(\mathsf R)],
\)
while Lemma~\ref{lem:saving} gives the averaged class-saving bound
\(
\overline S\ge \E[\mathcal W(\mathsf R)]-\Pi.
\)
The target inequality \eqref{eq:target} can therefore be proved by
combining the contributions of $h_\gamma$ and $\mathcal W$ for each
possible normalized size.  The following theorem formalizes this
reduction.

\begin{theorem}[Cross-shift master certificate]
\label{thm:master}
Fix $\rho>1$, a comparison level $\gamma\ge0$, and a target factor
$\alpha>1$.  Let $\psi$ be a scale profile with profile cost $\Pi$ and
unit-window function $\mathcal W$.  Suppose that there exists
$t\in\mathbb R$ such that
\begin{align}
\mathcal W(r)
+\alpha h_\gamma(r)
-(2-\alpha)r
&\ge t,
\qquad r>0,
\label{eq:master-pointwise}\\
t+\alpha(2\gamma-\gamma^2)-1-\Pi
&\ge 0.
\label{eq:master-scalar}
\end{align}
Then
\(
2Q+1-\overline S\le\alpha(Q+F),
\)
and hence the target prefix inequality \eqref{eq:target} holds.
\end{theorem}


The theorem is the point at which the original high-dimensional structure
disappears.  Configuration information is represented by the single moment
$F$, while the random geometric classes enter only through the one-dimensional
window function $\mathcal W$.  Hence the complete prefix comparison reduces
to the pointwise condition~\eqref{eq:master-pointwise} and the scalar
condition~\eqref{eq:master-scalar}.

\begin{proof}
Taking expectation of the pointwise inequality \eqref{eq:master-pointwise} with respect to
$\mathsf R$ and using $Q=\E[\mathsf R]$ gives
\(
\E[\mathcal W(\mathsf R)]
+\alpha\E[h_\gamma(\mathsf R)]
-(2-\alpha)Q
\ge t.
\)
By Lemma~\ref{lem:saving},
\(
\overline S
\ge
\E[\mathcal W(\mathsf R)]-\Pi,
\)
and therefore
\(
\E[\mathcal W(\mathsf R)]
\le
\overline S+\Pi.
\)
Substituting this upper bound into the preceding inequality yields
\(
\overline S+\Pi
+\alpha\E[h_\gamma(\mathsf R)]
-(2-\alpha)Q
\ge t.
\)

The scalar condition \eqref{eq:master-scalar} can be rewritten as
\(
t-\Pi
\ge
1-\alpha(2\gamma-\gamma^2).
\)
Combining these two inequalities gives
\(
\overline S
+\alpha\E[h_\gamma(\mathsf R)]
-(2-\alpha)Q
\ge
1-\alpha(2\gamma-\gamma^2),
\)
or, equivalently,
\begin{align*}
2Q+1-\overline S
&\le
\alpha\left(
Q+2\gamma-\gamma^2+\E[h_\gamma(\mathsf R)]
\right).
\end{align*}

Finally, Theorem~\ref{thm:variance} implies
\(
2\gamma-\gamma^2+\E[h_\gamma(\mathsf R)]
\le F.
\)
Hence
\(
2Q+1-\overline S
\le
\alpha(Q+F),
\)
which is exactly \eqref{eq:target}.
\end{proof}

Theorem~\ref{thm:master} is the bridge from the analytic bounds to the
explicit certificate.  To prove the desired approximation factor, it
remains to choose $\rho$, $\gamma$, $\alpha$, and a scale profile
$\psi$ for which the two conditions of the theorem can be verified.

\subsection{Verification of the \texorpdfstring{$1.3168$}{1.3168} certificate}
\label{sec:certificate}

We now certify the two conditions of Theorem~\ref{thm:master} with
\[
\rho=\frac{53}{25},
\qquad
\gamma=\frac{89}{100},
\qquad
\alpha=\frac{823}{625}=1.3168,
\qquad
t=-\frac{107}{500}=-0.214.
\]
The structural analysis leading to Theorem~\ref{thm:master} is entirely
analytic.  
The overall dependency structure, including the location of the
computer-assisted step, is summarized in \Cref{app:proof-map}.
Computer assistance is used only in this final step to verify
an explicit finite certificate.  This is related in spirit to Li's
computer-assisted analysis \citep[Section~4]{Li2025}, but the
verification problem is different.  Li verifies bounds for a
size-configuration optimization problem, whereas our preceding analysis
has already reduced the ratio proof to the cap condition for a scale
profile, the scalar condition~\eqref{eq:master-scalar}, and the
one-dimensional pointwise condition~\eqref{eq:master-pointwise}.

For reproducibility, Appendices~\ref{app:certificate} and~\ref{app:verifier}
provide both the complete rational certificate 
and a compact version of the verification program.
The search procedure is used only for discovery.
At a high level, it discretizes the logarithmic scale, searches jointly over
admissible profile heights and the parameters of the master certificate, and
then rationalizes a feasible profile.  The proof itself does not rely on the
search: correctness is established solely by the explicit rational profile
and the independent verification described below.

The selected values of $\rho$, $\gamma$, and $\alpha$ are certified
feasible parameters; we do not claim that they are globally optimal within
the master-certificate framework.

Let
\(
\Delta:=\frac1{80}.
\)
For each integer $\ell$, let $\psi_\ell$ be the value listed in
\Cref{tab:certificate}, and set $\psi_\ell=0$ for every unlisted
index.  Define the step profile
\(
\psi(s):=\psi_\ell\)
for
\(s\in[\ell\Delta,(\ell+1)\Delta).
\)
The profile has 232 nonzero cells, indexed by
$\ell\in \{-157,\ldots,74\}$.  Each listed six-decimal height is interpreted
as the exact rational number with denominator $10^6$.

We first check the parts of the certificate that can be evaluated
exactly.  Since every profile cell has width $1/80$,
\begin{align*}
\Pi
&=
\int_{-\infty}^{\infty}\psi(s)^2\,ds 
=
\frac1{80}\sum_{\ell=-157}^{74}\psi_\ell^2 
=
\frac{1{,}735{,}592{,}044{,}371}
     {20{,}000{,}000{,}000{,}000}
=
0.08677960221855.
\end{align*}
Also,
\(
2\gamma-\gamma^2=\frac{9879}{10000},
\)
and therefore the left-hand side of the scalar condition
\eqref{eq:master-scalar} is
\begin{align*}
t+\alpha(2\gamma-\gamma^2)-1-\Pi
&=
\frac{1{,}742{,}355{,}629}
     {20{,}000{,}000{,}000{,}000} 
>0.
\end{align*}
Hence the scalar condition is verified exactly.

We next verify that $\psi$ satisfies the scale-profile cap
\(
\psi(s)\le b(\rho^s)\),
\(s\in\mathbb R\).
For $s\ge0$, we have $\rho^s\ge1$ and hence $b(\rho^s)=1$;
all corresponding profile heights are below $1$.  For a negative cell
\(
s\in
\left[\frac{\ell}{80},\frac{\ell+1}{80}\right),
\)
the function $b(\rho^s)$ is increasing in $s$.  Its minimum on the cell
is therefore attained at the left endpoint, so it suffices to verify
\(
\psi_\ell
\le
1-\sqrt{1-\rho^{\ell/40}}\),
\(
\ell\in \{-157,\ldots,-1\}.
\)
Directed interval arithmetic certifies all of these inequalities.  The
smallest certified cap slack occurs at $\ell=-134$ and is greater than
\(
2.1049724246\times10^{-6}.
\)

It remains to verify the pointwise condition
\eqref{eq:master-pointwise}.  Introduce the logarithmic variable
\(
u:=\log_\rho r\),
so that
\(r=\rho^u,
\)
and define
\(
W(u)
:=
\mathcal W(\rho^u)
=
2\int_{u-1}^{u}\psi(s)\,ds.
\)
Since $\psi$ is constant on cells of width $\Delta$, the function $W$
is affine on each interval
\(
\left[q\Delta,(q+1)\Delta\right].
\)
More precisely, for
$u\in[q\Delta,(q+1)\Delta]$,
\begin{align*}
W(q\Delta)
&=
2\Delta\sum_{\ell=q-80}^{q-1}\psi_\ell,\\
W(u)
&=
W(q\Delta)
+
2(u-q\Delta)(\psi_q-\psi_{q-80}),
\end{align*}
where unlisted profile heights are zero.
To express the pointwise master condition in the variable $u$, define
\[
G(u)
:=
W(u)
+\alpha h_\gamma(\rho^u)
-(2-\alpha)\rho^u.
\]
Because the map $u\mapsto\rho^u$ is a bijection from $\mathbb R$ to
$(0,\infty)$, condition~\eqref{eq:master-pointwise} is exactly
\(
G(u)\ge t=-0.214\)
for every \(u\in\mathbb R.
\)
Using the definition of $h_\gamma$, we have
\[
G(u)
=
\begin{cases}
W(u)-(2-\alpha)\rho^u,
& \rho^u<\gamma,\\[1mm]
W(u)
+2(\alpha-1)\rho^u
+\alpha\gamma^2\rho^{-u}
-2\alpha\gamma,
& \rho^u\ge\gamma.
\end{cases}
\]

The profile $\psi$ is nonzero exactly on the cells indexed by
$\ell\in \{-157,\ldots,74\}$, where each cell is
\(
\left[\frac{\ell}{80},\frac{\ell+1}{80}\right).
\)
Hence the union of all nonzero cells is
\(
\bigcup_{\ell=-157}^{74}
\left[\frac{\ell}{80},\frac{\ell+1}{80}\right)
=
\left[-\frac{157}{80},\frac{75}{80}\right).
\)
Therefore, $\psi(s)=0$ outside
\(
\left[-\frac{157}{80},\frac{75}{80}\right).
\)
It follows that the unit window $[u-1,u]$ contains no point at which
$\psi$ is nonzero whenever
\(
u\le -\frac{157}{80}\)
or 
\(u\ge \frac{155}{80}.
\)
Consequently, $W(u)=0$ on the two exterior regions.
It therefore remains to verify only the finite interval
\(
\left[-\frac{157}{80},\frac{155}{80}\right].
\)
Since the verification grid has width $\Delta=1/80$, this interval is
covered by the cells
\(
\left[\frac q{80},\frac{q+1}{80}\right]\),
\(q\in \{-157,\ldots,154\}.
\)
Thus only these finitely many cells require computer-assisted interval
verification.

We now explain why the computer-assisted verification of these cells is
rigorous.  It is not based on evaluating $G$ at a finite set of sample
points.  For each rational subinterval $U$, directed interval
arithmetic computes an enclosure
\(
\mathbf G(U)
\)
such that
\(
G(u)\in\mathbf G(U)\)
for every \(u\in U.
\)
Therefore, if the lower endpoint of $\mathbf G(U)$ is strictly greater
than $t$, then
\(
G(u)>t\)
for every \(u\in U.
\)
If this test does not yet succeed, the verifier bisects $U$ at its exact
rational midpoint and checks the two resulting subintervals
recursively.  Successful termination therefore gives a finite covering
of the entire bounded verification region by intervals on which the
desired inequality holds uniformly.  This is an interval-enclosure
proof over a continuum of values, rather than a numerical sampling
argument.

The only additional care is required at the threshold
$\rho^u=\gamma$, where the formula for $h_\gamma$ changes.  For every
subinterval $U$, the verifier first computes an outward-rounded interval
enclosure of $\rho^u$ for all $u\in U$.  It uses the first branch of
$G$ only if this entire enclosure lies below the exact rational value
$\gamma=89/100$, and the second branch only if the entire enclosure
lies at or above $\gamma$.  If the enclosure straddles the threshold,
the verifier uses
\(
h_\gamma(\rho^u)\ge0
\)
and drops the nonnegative term
$\alpha h_\gamma(\rho^u)$.  Thus
\(
G(u)
\ge
W(u)-(2-\alpha)\rho^u
\)
remains a valid lower bound on the whole subinterval.  The verifier then
subdivides further whenever this weaker enclosure is insufficient.
Consequently, no proof step depends on a floating-point approximation
of $\log_\rho\gamma$.

For completeness, the two infinite tails require no interval
subdivision.  On the left tail,
\(
G(u)=-(2-\alpha)\rho^u,
\)
which is minimized at its right endpoint.  On the right tail,
\(
G(u)
=
2(\alpha-1)\rho^u
+\alpha\gamma^2\rho^{-u}
-2\alpha\gamma.
\)
The verifier certifies both the boundary value and positivity of the
derivative at the left endpoint of this tail; the derivative is
increasing there, so $G$ remains increasing thereafter.  These checks
establish the required inequality on both unbounded regions.

The verifier certifies all finite cells after checking 217,894
subintervals, with maximum subdivision depth 11, and proves
\(
G(u)>-0.214\)
for every \(u\in\mathbb R.
\)
Thus the pointwise condition~\eqref{eq:master-pointwise} holds on its
entire domain.  The verification results are summarized in
\Cref{tab:verify}.

\begin{table}[htbp]
\centering
\caption{Computer-assisted verification of the $1.3168$ certificate.}
\label{tab:verify}
\begin{tabular}{ll}
\toprule
Quantity & Certified value \\
\midrule
Nonzero profile cells
& 232, indices $-157,\ldots,74$ \\
Exact profile cost $\Pi$
& $1{,}735{,}592{,}044{,}371/
   20{,}000{,}000{,}000{,}000$ \\
Exact scalar slack
& $1{,}742{,}355{,}629/
   20{,}000{,}000{,}000{,}000$ \\
Smallest cap slack
& $>2.1049724246\times10^{-6}$ \\
Finite pointwise cells
& $q=-157,\ldots,154$ \\
Subintervals checked
& 217,894 \\
Maximum subdivision depth
& 11 \\
Pointwise condition
& $G(u)>-0.214$ for every $u\in\mathbb R$ \\
Exterior tails
& verified analytically \\
\bottomrule
\end{tabular}
\end{table}

\begin{theorem}
\label{thm:prefix-factor}
For $\rho=53/25$, the target prefix inequality~\eqref{eq:target} holds
with
\(\alpha=\frac{823}{625}=1.3168\)
for every machine and every Smith prefix.
\end{theorem}

\begin{proof}
The exact calculation verifies the scalar condition
\eqref{eq:master-scalar}.  The cap checks show that $\psi$ is a valid
scale profile, and the computer-assisted interval proof establishes the
pointwise condition~\eqref{eq:master-pointwise} for every $r>0$.
The conclusion therefore follows from Theorem~\ref{thm:master}.
\end{proof}

\subsection{From the prefix factor to the global approximation guarantee}
\label{sec:global}

We now complete the proof of Theorem~\ref{thm:main}.  

\begin{proof}[Proof of Theorem~\ref{thm:main}]
Let $y$ be the
configuration-LP solution supplied to Li's rounding algorithm.  By
Theorem~\ref{thm:prefix-factor}, after averaging over the random class shift,
the algorithmic bracket associated with every machine $i$ and every Smith
prefix is at most
\(
\alpha:=\frac{823}{625}=1.3168
\)
times the corresponding configuration-LP bracket.  Since every coefficient
$\sigma_{j^*}-\sigma_{j^*+1}$ in the prefix decomposition is nonnegative,
Theorem~\ref{thm:li-prefix} therefore gives
\(
\E[\cost_i(\ALG)]
\le
\alpha\sum_f y_{if}\cost_i(f)
\)
for every machine $i$.  Summing over the machines yields
\begin{equation}\label{eq:global-lp}
\E[\cost(\ALG)]
\le
\alpha\,\LP(y).
\end{equation}

For every $\eta>0$, the configuration LP can be solved within a factor
$1+\eta$ in time polynomial in the input size and $1/\eta$
\citep{SviridenkoWiese2013}.  Let $\LP^\star$ denote its optimum.  Since the
configuration LP is a relaxation,
\(
\LP^\star\le\OPT.
\)
Using a solution $y$ satisfying
\(
\LP(y)\le(1+\eta)\LP^\star
\)
in \eqref{eq:global-lp} gives
\(
\E[\cost(\ALG)]
\le
\alpha(1+\eta)\OPT.
\)
For a prescribed $\varepsilon>0$, choose
\(
\eta:=\frac{\varepsilon}{\alpha}
=\frac{625}{823}\varepsilon.
\)
Then
\(
\alpha(1+\eta)=\alpha+\varepsilon
=1.3168+\varepsilon.
\)

Li's iterative rounding runs in polynomial time for every fixed
$\rho>1$ by Theorem~\ref{thm:li-prefix}; here $\rho=53/25$ is a constant.
The configuration-LP approximation scheme is polynomial in the input size
and $1/\varepsilon$.  The certificate verifier is an offline proof artifact
and is not part of the scheduling algorithm.  Hence the complete algorithm
runs in randomized polynomial time.

\end{proof}

\section{Conclusion}
\label{sec:conclusion}

We have improved the approximation guarantee for
$R\|\sum_j w_jC_j$ from $1.36+\varepsilon$ to
$1.3168+\varepsilon$ while retaining Li's configuration-LP and
iterative-rounding framework.  The improvement comes from changing the
order of the analysis rather than strengthening the rounding primitive:
we fix and normalize a physical Smith prefix before averaging over the
random geometric shift.  This keeps the prefix distribution and
configuration moments common to all shifts and allows the class savings
to be aggregated through a one-dimensional logarithmic window.

Two structural ingredients make this comparison possible.  The
scale-dependent configuration bound retains information about the location
of normalized job sizes, while the cross-shift saving bound converts the
random geometric classes into a unit-window transform of a scale profile.
Their combination yields the master certificate of
Theorem~\ref{thm:master}.  The explicit rational profile in
\Cref{sec:certificate} then certifies the factor $1.3168$ over the full
continuum.

The result has several natural limitations.  It concerns the standard
machine-independent-weight model, does not strengthen Li's rounding
distribution itself, and does not claim that $1.3168$ is optimal within the
master-certificate framework.  Further improvements may come from richer
lower bounds for the capped class saving, combinations of several comparison
levels in the configuration bound, or stronger information from the
underlying rounding procedure.  More broadly, the fix-first,
shift-second viewpoint may be useful in other approximation analyses in
which randomly shifted geometric partitions act on a common underlying
object.

\section*{Acknowledgements}


\noindent\textbf{Research data}: No empirical data were used. The complete rational certificate and the Python verifier are provided in Appendices~\ref{app:certificate} and~\ref{app:verifier}, respectively.

\vspace{2pt}
\noindent\textbf{Declaration of AI technologies in the manuscript preparation process}: 
The authors used ChatGPT to refine the writing and improve the presentation of the verification program. All AI-assisted content and code were reviewed and edited by the authors, who take full responsibility for the final manuscript and results.

\vspace{2pt}
\noindent\textbf{Declarations of interest}: None

\printbibliography

\newpage
\appendix

\section{Proof structure}\label{app:proof-map}

\Cref{fig:proof-dependencies} records the main dependencies in the proof of
Theorem~\ref{thm:main}.  The diagram distinguishes the results inherited
from prior work, the analytic arguments introduced here, and the single
computer-assisted verification step.  It is intended only as a guide to
the proof; definitions and local calculations are omitted.

\begin{figure}[htbp]
\centering
\resizebox{0.7\linewidth}{!}{%
\begin{tikzpicture}[
    >=Latex,
    node distance=9mm and 16mm,
    box/.style={
        draw,
        rounded corners=2pt,
        align=center,
        minimum height=10mm,
        text width=4.35cm,
        inner sep=4pt,
        font=\small
    },
    inherited/.style={
        box,
        dashed
    },
    analytic/.style={
        box
    },
    computer/.style={
        box,
        double,
        double distance=1.1pt
    },
    result/.style={
        box,
        very thick
    },
    flow/.style={
        ->,
        thick
    }
]


\node[inherited] (li)
{
Li's prefix interface\\
Theorem~\ref{thm:li-prefix}\\[-1mm]
\eqref{eq:li-lp-prefix},
\eqref{eq:li-round-prefix}
};

\node[analytic, right=22mm of li] (norm)
{
Fixed-prefix normalization\\
\Cref{sec:normalize}\\[-1mm]
$\mathsf R,\;Q,\;F,\;\overline S$
};

\node[
    analytic,
    below=12mm of $(li)!0.5!(norm)$
] (target)
{
Normalized prefix comparison\\
\eqref{eq:target}
};

\draw[flow] (li) -- (target);
\draw[flow] (norm) -- (target);


\node[
    analytic,
    below left=13mm and 20mm of target
] (variance)
{
Scale-dependent configuration bound\\
Theorem~\ref{thm:variance}
};

\node[
    analytic,
    below right=8mm and 20mm of target
] (tangent)
{
Affine one-class bound\\
Lemma~\ref{lem:tangent}
};

\node[
    analytic,
    below=9mm of tangent
] (saving)
{
Cross-shift saving bound\\
Lemma~\ref{lem:saving}
};

\draw[flow] (tangent) -- (saving);


\node[
    analytic,
    below=28mm of target
] (master)
{
Cross-shift master certificate\\
Theorem~\ref{thm:master}\\[-1mm]
\eqref{eq:master-pointwise},
\eqref{eq:master-scalar}
};

\draw[flow] (target) -- (master);
\draw[flow] (variance) -- (master);
\draw[flow] (saving) -- (master);


\node[
    computer,
    below=10mm of master
] (cert)
{
Computer-assisted $1.3168$ certificate\\
\Cref{sec:certificate}\\[-1mm]
232-cell rational profile;\\
cap, scalar, and pointwise checks
};

\draw[flow]
    (master)
    --
    node[right,font=\scriptsize]
    {conditions to verify}
    (cert);


\node[
    result,
    below=10mm of cert
] (prefix)
{
Certified $1.3168$ prefix factor\\
Theorem~\ref{thm:prefix-factor}
};

\draw[flow] (cert) -- (prefix);


\node[
    analytic,
    below=11mm of prefix
] (global)
{
Prefix-to-global lifting\\
and running time\\
\Cref{sec:global}
};

\node[
    inherited,
    right=20mm of global
] (lpapprox)
{
$(1+\eta)$-approximate\\
configuration LP\\
\citep{SviridenkoWiese2013}
};

\draw[flow] (prefix) -- (global);
\draw[flow] (lpapprox) -- (global);

\draw[flow]
    (li.west)
    -- ++(-40mm,0)
    |- node[
        pos=0.72,
        left,
        font=\scriptsize,
        align=right
    ]
    {prefix decomposition\\
     and polynomial rounding}
    (global.west);


\node[
    result,
    below=11mm of global
] (main)
{
Randomized $(1.3168+\varepsilon)$-approximation\\
Theorem~\ref{thm:main}
};

\draw[flow] (global) -- (main);

\end{tikzpicture}%
}

\caption{
Proof dependencies for Theorem~\ref{thm:main}.
Dashed boxes denote results imported from prior work, ordinary boxes denote
analytic arguments proved here, and the double-bordered box denotes the
computer-assisted certificate verification.  Computer assistance enters
only after Theorem~\ref{thm:master}.
}

\label{fig:proof-dependencies}
\end{figure}
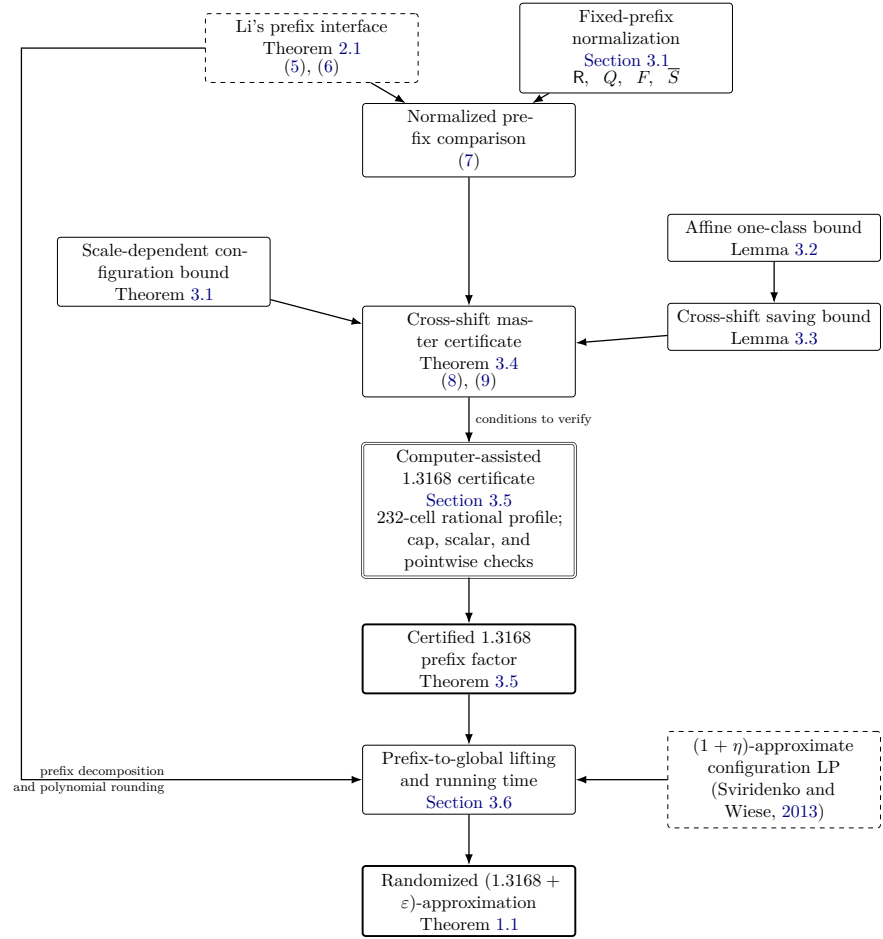

\newpage
\section{Complete rational certificate}\label{app:certificate}

\begin{table}[htbp]
\centering
\caption{Complete profile heights $\psi_\ell$. The profile equals $\psi_\ell$ on
$[\ell/80,(\ell+1)/80)$; unlisted indices have value zero.}
\label{tab:certificate}
\begingroup
\singlespacing
\setlength{\tabcolsep}{2.0pt}
\renewcommand{\arraystretch}{0.92}
\begin{adjustbox}{width=0.98\textwidth,center}
\begin{tabular}{*{8}{r@{\,:\,}l}}
\toprule
$\ell$ & $\psi_\ell$ & $\ell$ & $\psi_\ell$ & $\ell$ & $\psi_\ell$ & $\ell$ & $\psi_\ell$ & $\ell$ & $\psi_\ell$ & $\ell$ & $\psi_\ell$ & $\ell$ & $\psi_\ell$ & $\ell$ & $\psi_\ell$\\
\midrule
-157 & 0.000460 & -156 & 0.001835 & -155 & 0.003256 & -154 & 0.004723 & -153 & 0.006237 & -152 & 0.007801 & -151 & 0.009415 & -150 & 0.011082\\
-149 & 0.012802 & -148 & 0.014578 & -147 & 0.016411 & -146 & 0.018303 & -145 & 0.020255 & -144 & 0.022270 & -143 & 0.024350 & -142 & 0.026496\\
-141 & 0.028711 & -140 & 0.030996 & -139 & 0.033355 & -138 & 0.035790 & -137 & 0.038303 & -136 & 0.039636 & -135 & 0.040404 & -134 & 0.041187\\
-133 & 0.041985 & -132 & 0.042799 & -131 & 0.043629 & -130 & 0.044475 & -129 & 0.045339 & -128 & 0.046219 & -127 & 0.047117 & -126 & 0.048033\\
-125 & 0.048968 & -124 & 0.049921 & -123 & 0.050893 & -122 & 0.051884 & -121 & 0.052896 & -120 & 0.053927 & -119 & 0.054980 & -118 & 0.056053\\
-117 & 0.057148 & -116 & 0.058265 & -115 & 0.059405 & -114 & 0.060568 & -113 & 0.061754 & -112 & 0.062965 & -111 & 0.064200 & -110 & 0.065460\\
-109 & 0.066745 & -108 & 0.068057 & -107 & 0.069396 & -106 & 0.070762 & -105 & 0.072156 & -104 & 0.073579 & -103 & 0.075030 & -102 & 0.076512\\
-101 & 0.078024 & -100 & 0.079568 & -99 & 0.081143 & -98 & 0.082751 & -97 & 0.084393 & -96 & 0.086068 & -95 & 0.087779 & -94 & 0.089525\\
-93 & 0.091077 & -92 & 0.092466 & -91 & 0.094987 & -90 & 0.096885 & -89 & 0.098823 & -88 & 0.100802 & -87 & 0.102823 & -86 & 0.104887\\
-85 & 0.106995 & -84 & 0.109148 & -83 & 0.111347 & -82 & 0.113594 & -81 & 0.115889 & -80 & 0.118233 & -79 & 0.120629 & -78 & 0.123077\\
-77 & 0.125578 & -76 & 0.128134 & -75 & 0.130746 & -74 & 0.133416 & -73 & 0.136145 & -72 & 0.138935 & -71 & 0.141787 & -70 & 0.144703\\
-69 & 0.147684 & -68 & 0.150733 & -67 & 0.153851 & -66 & 0.157039 & -65 & 0.160301 & -64 & 0.163637 & -63 & 0.167051 & -62 & 0.170544\\
-61 & 0.174118 & -60 & 0.177776 & -59 & 0.181520 & -58 & 0.185353 & -57 & 0.189277 & -56 & 0.192036 & -55 & 0.194241 & -54 & 0.196476\\
-53 & 0.198739 & -52 & 0.201032 & -51 & 0.203356 & -50 & 0.205710 & -49 & 0.208095 & -48 & 0.210511 & -47 & 0.212960 & -46 & 0.215441\\
-45 & 0.218276 & -44 & 0.221156 & -43 & 0.224039 & -42 & 0.226924 & -41 & 0.229812 & -40 & 0.232704 & -39 & 0.235599 & -38 & 0.238498\\
-37 & 0.241401 & -36 & 0.244308 & -35 & 0.247220 & -34 & 0.250137 & -33 & 0.253059 & -32 & 0.255986 & -31 & 0.258919 & -30 & 0.261857\\
-29 & 0.264802 & -28 & 0.267753 & -27 & 0.270711 & -26 & 0.273675 & -25 & 0.276647 & -24 & 0.279626 & -23 & 0.282613 & -22 & 0.285608\\
-21 & 0.288611 & -20 & 0.291623 & -19 & 0.294643 & -18 & 0.297672 & -17 & 0.300711 & -16 & 0.303759 & -15 & 0.306817 & -14 & 0.309885\\
-13 & 0.312963 & -12 & 0.315357 & -11 & 0.311781 & -10 & 0.307600 & -9 & 0.303478 & -8 & 0.299415 & -7 & 0.295412 & -6 & 0.291468\\
-5 & 0.287586 & -4 & 0.283764 & -3 & 0.280004 & -2 & 0.276305 & -1 & 0.272670 & 0 & 0.269099 & 1 & 0.265591 & 2 & 0.262149\\
3 & 0.258772 & 4 & 0.255461 & 5 & 0.252219 & 6 & 0.249044 & 7 & 0.245939 & 8 & 0.242904 & 9 & 0.239941 & 10 & 0.237050\\
11 & 0.234233 & 12 & 0.231491 & 13 & 0.228824 & 14 & 0.226236 & 15 & 0.223727 & 16 & 0.221298 & 17 & 0.218951 & 18 & 0.216688\\
19 & 0.214511 & 20 & 0.212420 & 21 & 0.210420 & 22 & 0.208510 & 23 & 0.206694 & 24 & 0.203713 & 25 & 0.200181 & 26 & 0.196678\\
27 & 0.193204 & 28 & 0.189759 & 29 & 0.186344 & 30 & 0.182957 & 31 & 0.179600 & 32 & 0.176271 & 33 & 0.172972 & 34 & 0.169702\\
35 & 0.166782 & 36 & 0.163902 & 37 & 0.161019 & 38 & 0.158134 & 39 & 0.155245 & 40 & 0.152354 & 41 & 0.149458 & 42 & 0.146559\\
43 & 0.143656 & 44 & 0.140749 & 45 & 0.137837 & 46 & 0.134920 & 47 & 0.131998 & 48 & 0.129071 & 49 & 0.126138 & 50 & 0.123200\\
51 & 0.120255 & 52 & 0.117304 & 53 & 0.114346 & 54 & 0.111381 & 55 & 0.108409 & 56 & 0.105430 & 57 & 0.102443 & 58 & 0.099448\\
59 & 0.096445 & 60 & 0.093434 & 61 & 0.090413 & 62 & 0.087384 & 63 & 0.084346 & 64 & 0.081297 & 65 & 0.078240 & 66 & 0.075171\\
67 & 0.072093 & 68 & 0.068309 & 69 & 0.058533 & 70 & 0.048130 & 71 & 0.037763 & 72 & 0.027432 & 73 & 0.017136 & 74 & 0.006875\\
\bottomrule
\end{tabular}
\end{adjustbox}
\endgroup
\end{table}

\newpage
\section{Certificate verifier}\label{app:verifier}

\begingroup
\singlespacing

\setlength{\columnsep}{12pt}
\setlength{\columnseprule}{0.2pt}
\setlength{\multicolsep}{2pt}
\setlength{\premulticols}{0pt}
\setlength{\postmulticols}{0pt}

\raggedcolumns

\begin{multicols}{2}
\begin{lstlisting}[
  style=appendixpython2col
]
"""Compact rigorous verifier for the 1.3168 cross-shift certificate.
Requires Python 3 and mpmath 1.3.0.  All proof-critical tests use require(),
so they remain active under python -O.  The final minimization is diagnostic
only; the proof is the exact/interval verification preceding it.
"""
from fractions import Fraction as F
from functools import lru_cache
import mpmath as mp

mp.iv.dps=70; mp.mp.dps=100
def require(c,m):
    if not c: raise RuntimeError(m)
def IV(q): return mp.iv.mpf(q.numerator)/q.denominator
def MP(q): return mp.mpf(q.numerator)/q.denominator

RHO_Q,ALPHA_Q,GAMMA_Q,H_Q,T_Q=F(53,25),F(823,625),F(89,100),F(1,80),F(-107,500)
require(RHO_Q>1 and ALPHA_Q>1 and GAMMA_Q>=0 and H_Q>0,"bad parameters")
RHO,ALPHA,GAMMA,H,T=map(IV,(RHO_Q,ALPHA_Q,GAMMA_Q,H_Q,T_Q))
MICRO=10**6

VALUES="""0.000460 0.001835 0.003256 0.004723 0.006237 0.007801 0.009415 0.011082
0.012802 0.014578 0.016411 0.018303 0.020255 0.022270 0.024350 0.026496
0.028711 0.030996 0.033355 0.035790 0.038303 0.039636 0.040404 0.041187
0.041985 0.042799 0.043629 0.044475 0.045339 0.046219 0.047117 0.048033
0.048968 0.049921 0.050893 0.051884 0.052896 0.053927 0.054980 0.056053
0.057148 0.058265 0.059405 0.060568 0.061754 0.062965 0.064200 0.065460
0.066745 0.068057 0.069396 0.070762 0.072156 0.073579 0.075030 0.076512
0.078024 0.079568 0.081143 0.082751 0.084393 0.086068 0.087779 0.089525
0.091077 0.092466 0.094987 0.096885 0.098823 0.100802 0.102823 0.104887
0.106995 0.109148 0.111347 0.113594 0.115889 0.118233 0.120629 0.123077
0.125578 0.128134 0.130746 0.133416 0.136145 0.138935 0.141787 0.144703
0.147684 0.150733 0.153851 0.157039 0.160301 0.163637 0.167051 0.170544
0.174118 0.177776 0.181520 0.185353 0.189277 0.192036 0.194241 0.196476
0.198739 0.201032 0.203356 0.205710 0.208095 0.210511 0.212960 0.215441
0.218276 0.221156 0.224039 0.226924 0.229812 0.232704 0.235599 0.238498
0.241401 0.244308 0.247220 0.250137 0.253059 0.255986 0.258919 0.261857
0.264802 0.267753 0.270711 0.273675 0.276647 0.279626 0.282613 0.285608
0.288611 0.291623 0.294643 0.297672 0.300711 0.303759 0.306817 0.309885
0.312963 0.315357 0.311781 0.307600 0.303478 0.299415 0.295412 0.291468
0.287586 0.283764 0.280004 0.276305 0.272670 0.269099 0.265591 0.262149
0.258772 0.255461 0.252219 0.249044 0.245939 0.242904 0.239941 0.237050
0.234233 0.231491 0.228824 0.226236 0.223727 0.221298 0.218951 0.216688
0.214511 0.212420 0.210420 0.208510 0.206694 0.203713 0.200181 0.196678
0.193204 0.189759 0.186344 0.182957 0.179600 0.176271 0.172972 0.169702
0.166782 0.163902 0.161019 0.158134 0.155245 0.152354 0.149458 0.146559
0.143656 0.140749 0.137837 0.134920 0.131998 0.129071 0.126138 0.123200
0.120255 0.117304 0.114346 0.111381 0.108409 0.105430 0.102443 0.099448
0.096445 0.093434 0.090413 0.087384 0.084346 0.081297 0.078240 0.075171
0.072093 0.068309 0.058533 0.048130 0.037763 0.027432 0.017136 0.006875""".split()
require(len(VALUES)==232,"expected 232 profile heights")
def micro(s):
    q=F(s)*MICRO; require(q.denominator==1,f"non-micro height: {s}"); return q.numerator
NUM={j:micro(s) for j,s in zip(range(-157,75),VALUES)}
require(all(v>=0 for v in NUM.values()),"negative profile height")
JMIN,JMAX=min(NUM),max(NUM)
require((JMIN,JMAX)==(-157,74) and set(NUM)==set(range(JMIN,JMAX+1)),"bad support")
PSI={j:mp.iv.mpf(s) for j,s in zip(range(-157,75),VALUES)}
def psi(j): return PSI.get(j,mp.iv.mpf(0))
def psi_q(j): return F(NUM[j],MICRO) if j in NUM else F(0)

# Exact profile cost and scalar condition.
P=F(sum(v*v for v in NUM.values()),80*MICRO*MICRO)
B=2*GAMMA_Q-GAMMA_Q*GAMMA_Q
SLACK=T_Q+ALPHA_Q*B-1-P
require(P==F(1735592044371,20000000000000),"wrong Pi")
require(SLACK==F(1742355629,20000000000000) and SLACK>0,"scalar condition failed")

# Scale-profile cap.
CAP_TEXT="2.1049724246e-6"; CAP=mp.iv.mpf(CAP_TEXT); CAP_MIN=-134
cs={}
for j in range(JMIN,0):
    tau=RHO**(mp.iv.mpf(j)/80); s=1-mp.iv.sqrt(1-tau*tau)-psi(j)
    require(s.a>CAP.b,f"cap failed at {j}: {s}"); cs[j]=s
for j in range(0,JMAX+1): require((1-psi(j)).a>0,f"cap failed at {j}")
require(all(s.a>cs[CAP_MIN].b for j,s in cs.items() if j!=CAP_MIN),"cap minimum not unique")

# W(u)=2 int_[u-1,u] psi(s) ds and pointwise interval enclosure.
@lru_cache(None)
def W0(k): return 2*H*sum((psi(j) for j in range(k-80,k)),mp.iv.mpf(0))
def Giv(k,a,b):
    A,B=IV(a),IV(b); U=mp.iv.mpf([A.a,B.b]); left=mp.iv.mpf(k)/80
    W=W0(k)+2*(U-left)*(psi(k)-psi(k-80)); r=RHO**U
    if r.b<GAMMA.a: return W-(2-ALPHA)*r
    if r.a>=GAMMA.b: return W+2*(ALPHA-1)*r+ALPHA*GAMMA*GAMMA/r-2*ALPHA*GAMMA
    return W-(2-ALPHA)*r

count=depthmax=0
def certify(k,a,b,d=0):
    global count,depthmax
    count+=1; depthmax=max(depthmax,d); E=Giv(k,a,b)
    if E.a>T: return
    require(d<45,f"uncertified cell {k} [{a},{b}]: {E}")
    c=(a+b)/2; certify(k,a,c,d+1); certify(k,c,b,d+1)

KMIN,KMAX=JMIN,JMAX+80
for k in range(KMIN,KMAX+1): certify(k,F(k,80),F(k+1,80))

# Infinite tails.
ul=mp.iv.mpf(JMIN)/80
require((-(2-ALPHA)*RHO**ul).a>T,"left tail failed")
ur=mp.iv.mpf(JMAX+81)/80; rr=RHO**ur
Gr=2*(ALPHA-1)*rr+ALPHA*GAMMA*GAMMA/rr-2*ALPHA*GAMMA
Dr=mp.iv.log(RHO)*(2*(ALPHA-1)*rr-ALPHA*GAMMA*GAMMA/rr)
require(Gr.a>T and Dr.a>0,"right tail failed")

# Non-proof high-precision diagnostic of the true piecewise-analytic minimum.
rho,alpha,gamma,t=map(MP,(RHO_Q,ALPHA_Q,GAMMA_Q,T_Q)); L=mp.log(rho); ug=mp.log(gamma)/L
def pmp(j): return MP(psi_q(j))
@lru_cache(None)
def w0(k): return 2*MP(H_Q)*sum((pmp(j) for j in range(k-80,k)),mp.mpf("0"))
def Wmp(u,k): return w0(k)+2*(pmp(k)-pmp(k-80))*(u-mp.mpf(k)/80)
def Gmp(u,k):
    r=rho**u; h=(r-gamma)**2/r if r>=gamma else mp.mpf("0")
    return Wmp(u,k)+alpha*h-(2-alpha)*r
def critical(k):
    m=2*(pmp(k)-pmp(k-80)); A=2*(alpha-1); b=m/L
    r=(-b+mp.sqrt(b*b+4*A*alpha*gamma*gamma))/(2*A)
    return mp.log(r)/L
best=None
def test(u,k):
    global best
    g=Gmp(u,k)
    if best is None or g<best[0]: best=(g,u,k)
for k in range(KMIN,KMAX+1):
    lo,hi=mp.mpf(k)/80,mp.mpf(k+1)/80
    if lo<ug: test(lo,k); test(min(hi,ug),k)
    if hi>=ug:
        a=max(lo,ug); test(a,k); test(hi,k); u=critical(k)
        if a<=u<=hi: test(u,k)

gmin,umin,qmin=best
print(f"cells={len(NUM)}; support={JMIN},...,{JMAX}")
print(f"Pi={P.numerator}/{P.denominator} = {float(P):.15f}")
print(f"scalar slack={SLACK.numerator}/{SLACK.denominator} = {float(SLACK):.15g}")
print(f"minimum cap slack > {CAP_TEXT} at j={CAP_MIN}")
print(f"pointwise cells q={KMIN},...,{KMAX}; subsegments={count}; max depth={depthmax}")
print("G(u) > -0.214 for all real u; both tails certified")
print("CERTIFICATE STATUS: VERIFIED")
print("NON-PROOF DIAGNOSTIC:")
print("  q =",qmin)
print("  u =",mp.nstr(umin,30))
print("  G(u) =",mp.nstr(gmin,30))
print("  margin =",mp.nstr(gmin-t,20))
\end{lstlisting}
\end{multicols}

\end{document}